%% file: ACC2018_Unmatched_Uncertainty.tex
\documentclass[letter, 10pt, conference]{ieeeconf}
\IEEEoverridecommandlockouts
\usepackage{epsfig}
\usepackage{subfigure}
\usepackage{graphicx}
\usepackage{amsfonts}
\usepackage{amsmath}
\usepackage{amssymb}
\usepackage{color}
\usepackage{pstricks}
\usepackage{url}
\usepackage{cite}
\usepackage{longtable}
\usepackage{algorithm}
\usepackage{algorithmic}
\newtheorem{theorem}{Theorem}

\usepackage{tikz}

\usepackage[normalem]{ulem}

\title{Hybrid Direct-Indirect Adaptive Control of Nonlinear System with Unmatched Uncertainty}

\author{Girish Joshi and Girish Chowdhary% <-this % stops a space
	\thanks{*This work was supported by FA9550-14-1-0399, FA9550-15-1-0146.}% <-this % stops a space
	\thanks{Authors are with Coordinated Science Laboratory and Department of Aerospace Engineering, University of Illinois,
		Urbana-Champaign, IL, USA
		{\tt\small girishj2@illinois.edu,girishc@illinois.edu}}%
}

\begin{document}
	
	\maketitle \thispagestyle{empty} \pagestyle{empty}

\begin{abstract}

\input{abstract}

\end{abstract}

\input{Introduction}

\input{System_Description}

\input{Controller_overview}

\input{Adaptive_Identification}

\input{Hybrid_MRAC}

\input{Simulations}

\input{conclusion}

\bibliographystyle{unsrt}

\bibliography{ACC2018_Unmatched_Uncertainty}

\end{document}

%% file: abstract.tex
%In this paper, we present a hybrid direct-indirect adaptive
%control architecture using MRAC to address a class of problems with matched and unmatched uncertainties. We present an architecture where the direct adaptive controller cancels the matched uncertainty and achieves reference model tracking for the nominal plant model. The unmatched uncertainty is estimated on-line and upon convergence is remodeled into a state dependent linear form to augment the nominal system dynamics, for direct adaptive control synthesis. The estimation of unmatched uncertainty is accomplished with observing the tracking error alone. We demonstrate that the proposed method can handle a wide class of nonlinear uncertainties for both matched and unmatched uncertainties on a non-linear plant model.
In this paper, we present a hybrid direct-indirect model reference adaptive controller (MRAC), to address a class of problems with matched and unmatched uncertainties. In the proposed architecture, the unmatched uncertainty is estimated online through a companion observer model. Upon convergence of the observer, the unmatched uncertainty estimate is remodeled into a state dependent linear form to augment the nominal system dynamics. Meanwhile, a direct adaptive controller designed for a switching system cancels the effect of matched uncertainty in the system and achieves reference model tracking. We demonstrate that the proposed hybrid controller can handle a broad class of nonlinear systems with both matched and unmatched uncertainties.

%% file: Introduction.tex
\section{Introduction}
Direct adaptive controllers for matched uncertainty, i.e., uncertainties in the span of the control, have been studied widely in the adaptive control literature  \cite{aastrom2013adaptive}. The direct adaptive controllers work exceedingly well for a large class of uncertain dynamical systems
%with modeling error and disturbances
and have been theoretically and experimentally proven to outperform non-adaptive baselines \cite{chowdhary2012experimental,tao2003adaptive,leman2009l1}.

However, the matched uncertainty condition is restrictive for the generalization of adaptive control to a broader class of nonlinear systems. Notably, the direct adaptive controllers cannot handle the systems with uncertainties appearing in the levels-of-differentiation other than that of the control input; i.e., not in the span of the control \cite{ 653051}. The objective of this paper is to provide hybrid adaptive control methods that can handle a broad class of systems with both matched and unmatched uncertainties. 

We present a hybrid adaptive control architecture where an estimate of unmatched uncertainty is used directly in the controller synthesis for the plant. We achieve this by augmenting the nominal system model with the estimate of unmatched uncertainty. Further, we use this augmented model to synthesize the total direct-indirect adaptive controller to achieve the desired reference model tracking.

\subsection{State of the Art}

While there is a plethora of work available in the literature handling matched uncertainties, there are very few references to controllers specifically focused on addressing the systems with unmatched uncertainties. A $\mathcal{L}_1$-adaptive control architecture for a class of nonlinear systems with unmatched uncertainties is presented in \cite{xargay2010l1}. An adaptive control using a back-stepping approach for observable minimum phase nonlinear systems with unmatched uncertainty is presented in \cite{koshkouei2000adaptive}. The combination of integral sliding-mode control and other robust techniques like $\mathcal{H_1}$ is examined in \cite{castanos2006analysis}. The formulation in \cite{yang2010lmi} addresses the application of a Linear Matrix Inequality (LMI)-based tool for analyzing the stability characteristics and the performance degradation of an adaptive system in the presence of unmatched uncertainties. While the methods \cite{castanos2006analysis}, \cite{yang2010lmi} have been shown to guarantee the stability in the presence of unmatched uncertainty, they do not use the unmatched uncertainty estimates directly for control generation, but rather ensure robustness in their presence. 

A hybrid model reference adaptive control for unmatched uncertainties is presented in \cite{quindlen2015hybrid}. The controller proposed in \cite{quindlen2015hybrid} utilizes the Concurrent Learning MRAC architecture. The premise of this work is a hybrid structure in identification law; learning the parameterization of matched and unmatched uncertainty simultaneously. As the unmatched uncertainty estimation error falls below a determinable and sufficiently small bound, these estimates are used to recompute the reference model. Meanwhile, a direct adaptive control architecture ensures the desired performance of the system. The controller presented in \cite{xargay2010l1},\cite{quindlen2015hybrid}  are restricted to the domain of linear in state uncertainties. This assumption is very restrictive and narrows the analysis to only a class of structured and linearly parametrized uncertainties. Also in \cite{quindlen2015hybrid}, the unmatched uncertainty estimation is based on system acceleration $\dot x(t)$ information which is usually not available directly for dynamical systems.

\subsection{Main Contribution}
We propose an architecture where online learned estimates of unmatched uncertainty are used in synthesizing the total adaptive controller to achieve desired reference model tracking.    
We leverage the fact that matched and unmatched uncertainties are in the mutual null space and do not corrupt the mutual information in the tracking error. The parametric estimates of both matched and unmatched uncertainties are captured using only tracking error information $e(t) = x(t) - \hat x(t)$ and unlike \cite{quindlen2015hybrid} we do not need the information of $\dot x(t)$ for unmatched weight update.
The controller presented in this paper caters to a broad class of nonlinear systems with matched and unmatched uncertainties. We treat both uncertainties as unstructured and thereby any generic neural network representation with non-linear basis functions $\phi(x)$ is admissible.

%% file: System_Description.tex
\section{System Description}
\label{system_description}
The dynamical system considered for presenting the hybrid direct-indirect adaptive control architecture is 
\begin{equation}
\dot x(t) = f(x(t)) + \mathcal{B}u(t)
\label{eq:0}
\end{equation}

In the above equation, the function $f: \mathbb{R}^n \to \mathbb{R}^n$ is assumed to be Lipschitz continues in $x,\dot x \in \mathcal{D}_x$. Let $ \mathcal{D}_x \subset \mathbb{R}^n$ be compact and the control $u(t)$ is assumed to belong to a set of admissible control inputs of measurable and bounded functions, ensuring the existence and uniqueness of the solution to \eqref{eq:0}.

By adding and subtracting $Ax(t)$, the nonlinear control affine model \eqref{eq:0} can be written in terms of the designed nominal plant model as 
\begin{eqnarray}
\dot x(t) &=& Ax(t) + \mathcal{B}u(t)+ (f(x(t)) - Ax(t)) \\
\dot x(t) &=& Ax(t) + \mathcal{B}\left(u(t) + \Delta_m(x)\right)+\mathcal{B}_u \Delta_{u}(x)
\label{eq:1}
\end{eqnarray}

Where $x(t) \in \mathcal{D}_x \subset \mathbb{R}^n$ is the state vector of the system and $u(t) \in \mathbb{R}^m$ is admissible full state feed-back control. $A \in \mathbb{R}^{n \times n}$ and $\mathcal{B} \in \mathbb{R}^{n \times m}$ are  Linear Time Invariant (LTI) nominal system matrices assumed to be known and the pair $\left(A,\mathcal{B}\right)$ is controllable. It is also assumed, the complete state information is available through system output $y(t) = x(t)$. 

The uncertainties, $\Delta_m(x) \in \mathbb{R}^m$, $\Delta_{u}(x)\in \mathbb{R}^{m_2}$  are the matched and unmatched components of the model approximation error $\Delta(x) = f(x(t)) - Ax(t)$.
The term, $\Delta_m(x)$ represents matched part of uncertainty and is within the span of input matrix $\mathcal{B}$, whereas $\Delta_{u}(x)$ represents the unmatched uncertainty and lies outside of the span of $\mathcal{B}$ (i.e., in the null-space $\mathcal{B}$). Hence, the control vector $u(t)$ is unable to cancel the effect of unmatched uncertainty on to the system \cite{narendra1978direct,narendra1979direct}. As a result, the typical direct adaptive control approach that relies on cancellation of the uncertainty through an adaptive element to achieve the desired reference behavior of the system are often found inadequate \cite{hovakimyan2010ℒ1,tao2003adaptive,narendra2012stable,lavretsky2012robust,narendra1986robust}.

\textbf{\emph{Remark 1:}} From the above statement, it can be inferred the matrix $\mathcal{B}_u \in \mathbb{R}^{n \times m_2}$ resides in the left null space of $\mathcal{B}$, i.e $\mathcal{B},\mathcal{B}_u$ are mutually orthogonal complements of each other, such that $\mathcal{B}_u^T\mathcal{B} = 0$.

From Remark-1, using the property of $\mathcal{B}$ and $\mathcal{B}_u$, we can define two full column rank matrices, $\mathcal{R}$ and $\mathcal{N}$ s.t $\mathcal{R}$ spans range space $\mathcal{R}(\mathcal{B})$ and $\mathcal{N}$ spans left null space $\mathcal{N}(\mathcal{B}^T)$ of $\mathcal{B}$. Hence, the matrices $\mathcal{B}$ and $\mathcal{B}_u$ can be represented as linear combinations of component vectors of $\mathcal{R}$ and $\mathcal{N}$ as follows,
\begin{equation}
\mathcal{B} = \mathcal{R}\alpha, \hspace{2mm} \mathcal{B}_u = \mathcal{N}\beta
\end{equation}

We can re-write the system dynamics (\ref{eq:1}) using projections of $\mathcal{B}$ and $\mathcal{B}_u$ onto the range space and left null space of $\mathcal{B}$ as follows
\begin{eqnarray}
\dot x(t) &=& Ax(t) + \mathcal{B}u(t) + \mathcal{R}\alpha\Delta_m(x) + \mathcal{N}\beta\Delta_{u}(x)\\
\dot x(t) &=& Ax(t) + \mathcal{B}u(t) + \left[\begin{array}{cc}
\mathcal{R} & \mathcal{N}
\end{array}\right]\left[\begin{array}{c}
\alpha\Delta_m(x)\\ 
\beta\Delta_{u}(x)
\end{array} \right]
\end{eqnarray}
Defining the matrix $\Omega = \left[\mathcal{R} \hspace{2mm} \mathcal{N}\right] \in \mathbb{R}^{n \times n}$, we can note that ``$\Omega$" from definition of $\mathcal{R},\mathcal{N}$, is a full column rank matrix, i.e. $Rank(\Omega) = n$

The system dynamics with joint representation of uncertainty can be written as follows,
\begin{equation}
\dot x(t) = Ax(t) + \mathcal{B}u(t) + \Delta(x)
\label{eq:unified uncertainty dynamics}
\end{equation}
where $$\Delta(x) =  \Omega\left[\begin{array}{c}
\alpha\Delta_m(x)\\ 
\beta\Delta_{u}(x)
\end{array} \right]$$

%% file: Controller_overview.tex
\section{Overview of the Presented Method}
The presented hybrid architecture of the controller for the system with matched and unmatched uncertainties uses the combination of both direct and indirect adaptive controllers. This hybrid direct-indirect approach is required since a direct adaptive controller alone cannot cancel the effect of unmatched uncertainty.  The total controller is realized through a two-step process (1) Learning observer model, such that $\hat x(t) \to x(t)$ and (2) Reference model tracking as $x(t) \to x_{rm}(t)$. 

The goal of the direct adaptive controller is to enforce the uncertain system to track a stable reference model, that characterizes the desired closed-loop response. The details of the direct adaptive controller are provided in Section \ref{direct}.

The reference model is assumed to be linear and therefore the desired transient and steady-state performance is defined by a selecting the system eigenvalues in the negative half plane. The desired closed-loop response of the reference system is given by
\begin{equation}
\dot x_{rm}(t) = A_{rm}x_{rm}(t) + B_{rm}r(t)
\label{eq:ref model}
\end{equation}
where $x_{rm}(t) \in \mathcal{D}_x  \subset \mathbb{R}^{n}$ and $A_{rm} \in {\rm I\!R}^{n \times n}$ is Hurwitz. Furthermore, the command $r(t)$ denotes a bounded, piece-wise continuous, reference signal and we assume the reference model (\ref{eq:ref model}) is bounded input-bounded output (BIBO) stable \cite{ioannou1988theory}.

The observer plant model provides the estimates of the uncertainties in the system. The details of the companion observer plant model are provided in Section \ref{observer_dynamics}. The $\ell_2$-norm of the observer tracking error is used as a measure of confidence on the estimation of total uncertainty. Upon convergence of the observer, that is as the tracking error falls below a determinable sufficiently small threshold ``$\gamma$", the unmatched part of the total uncertainty is remodeled into state dependent coefficient (SDC) \cite{cloutier1997state} form to augment the nominal dynamics. The details of the SDC formulation of unmatched uncertainty is given in Section \ref{SDC}. Further, this augmented model is used for the synthesis of the direct-indirect adaptive controller for the nonlinear system to ensure the required reference model tracking.

%% file: Adaptive_Identification.tex
\section{Adaptive Identification}
\label{adaptive_identification}
This section provides details of the companion observer model and adaptive identification law for estimating the total uncertainty present in the system. The unmatched terms learned online using a system identifier approach are accounted for, in the total controller for the reference model tracking (Indirect Adaptive Controller). We show the guaranteed boundedness of the observer tracking errors near zero solution and network parameters under adaptive identification law \eqref{eq:18} 

\subsection{System Observer Model}
\label{observer_dynamics}
Based on the system dynamics \eqref{eq:unified uncertainty dynamics}, consider a Luenberger state observer, of the form.
\begin{eqnarray}
\dot{\hat x}(t) &=& A\hat x(t) + Bu(t) + \hat \Delta(x) + L_{\tau}\left(x(t)-\hat x(t)\right) \label{eq:6}
\end{eqnarray}
where $\hat x(t) \in  \mathbb{R}^n$ is state of the observer model. The term $\hat \Delta(x) \in  \mathbb{R}^n$ in \eqref{eq:6} represent the estimate of the total uncertainty in \eqref{eq:unified uncertainty dynamics} and
$L_{\tau}$ is the observer feedback gain. This feedback term helps placing the poles of observer tracking error dynamics to desired location further away from poles of reference plant, to make the observer tracking error dynamics faster \cite{lavretsky2012robust}. This condition is particularly helpful if using the observer information in the control synthesis.

The true uncertainty $\Delta(x)$ in unknown, but it is assumed to be continuous over a compact domain $\mathcal{D}_x \subset \mathbb{R}^n$. Neural Networks (NN) have been widely used to represent unstructured uncertainties whose basis is not known.  Using NN, the network estimate of the uncertainty can be written as  
\begin{equation}
\hat \Delta(x) \triangleq \hat W^T\phi(x)
\label{eq:nn_estimate_defn}
\end{equation}
where $\hat W \in \mathbb{R}^{k \times n}$ are network weights and $\phi(x) = [1,\phi_1(x),\phi_1, \ldots,\phi_k(x)]^T$ is a $k$ dimensional vector of chosen basis function. The basis vector $\phi(x)$ is considered to be Lipschitz continuous to ensure the existence and uniqueness of the solution (\ref{eq:unified uncertainty dynamics}).

From definition of the structure of total uncertainty (\ref{eq:unified uncertainty dynamics}) the estimate of the individual components of matched and unmatched uncertainty can be expressed as,
\begin{equation}
\Omega\left[\begin{array}{c}
\alpha \hat \Delta_m(x)\\ 
\beta \hat \Delta_{u}(x)
\end{array} \right] \triangleq \hat \Delta(x) = \hat W^T\phi(x)
\end{equation}

\begin{equation}
\Rightarrow \left[\begin{array}{c}
\alpha \hat \Delta_m(x)\\ 
\beta \hat \Delta_{u}(x)
\end{array} \right] = \Omega^{-1}\hat W^T\phi(x) = \xi(x)
\end{equation}

Therefore the estimate of matched and unmatched uncertainty are as follows,
\begin{eqnarray}
\hat \Delta_m(x) &=& \alpha^\dagger\xi_m(x)\\
\hat \Delta_u(x) &=& \beta^\dagger\xi_u(x)
\end{eqnarray}
where $\alpha^\dagger,\beta^\dagger$ are the left pseudo-inverse of $\alpha,\beta$. The projection of uncertainty on the range and the null space of $\mathcal{B}$ is represented by $\xi(x) = [\xi_m(x), \xi_u(x)]^T \in \mathbb{R}^n$ where $\xi_m(x) \in {\rm I\!R}^{r}$, $\xi_u(x) \in  \mathbb{R}^{n-r}$ where $r$ being column rank of $\mathcal{B}$

Appealing to the universal approximation property of NN \cite{park1991universal} 
we have that, given a fixed number of basis
functions $\phi(x) \in \mathbb{R}^k$ there exists ideal weights $W^* \in \mathbb{R}^{k \times n}$ and $\epsilon(x) \in \mathbb{R}^{n}$ such that the following
approximation holds
\begin{equation}
\Delta(x) = W^{*T}\phi(x) + \epsilon(x), \hspace{2mm} \forall x(t) \in \mathcal{D}_x \subset \mathbb{R}^{n}
\label{eq:3}
\end{equation}
The network approximation error $\epsilon(x)$ is upper bounded, s.t $\bar \epsilon = \sup_{x \in \mathcal{D}_x}\|\epsilon(x)\|$, and can be made arbitrarily small given sufficiently large number of basis functions.

\textbf{\emph{Assumption 1:}} For uncertainty parameterized by unknown true weight $W^* \in \mathbb{R}^{k \times n}$ and known nonlinear basis $\phi(x)$,
the ideal weight matrix is assumed to be upper bounded s.t $\|W^*\| \leq \mathcal{W}$.

Substituting \eqref{eq:nn_estimate_defn} in \eqref{eq:6}, the observer plant can be written as
\begin{eqnarray}
\dot{\hat x}(t) &=& A\hat x(t) + Bu(t) + \hat W^T\phi(x) + L_{\tau}\left(x(t)-\hat x(t)\right) \label{eq:7} \nonumber \\
\end{eqnarray}
The observer model tracking error is defined as 
\begin{equation}
e(t) = x(t)-\hat x(t) 
\end{equation}
Using (\ref{eq:unified uncertainty dynamics}, \ref{eq:6}) the tracking error dynamics can be written as
\begin{equation}
\dot e(t) = \dot x(t) - \dot{\hat x}(t)
\label{eq:13}
\end{equation}
\begin{equation}
\dot e(t) = \left(A - L_{\tau}\right)e(t) + \tilde W^T\phi(x) +  \epsilon(x)
\label{eq:14}
\end{equation}
where %$C = I_{n \times n}$ and
$L_{\tau} = diag\left[l_{\tau 1} \ldots l_{\tau n}\right] \in \mathbb{R}^{n \times n}$ is the tracking error feed-back gain in (\ref{eq:7}). Hence the observer tracking error dynamics can be written as
\begin{equation}
\dot e(t) = A_\tau e(t) + \tilde W^T\phi(x) +  \epsilon(x)
\label{eq:15}
\end{equation}
where $A_\tau = \left(A - L_{\tau}\right)$ is Hurwitz s.t $\lambda_{min}(A_\tau) < \lambda_{min}(A_{rm})$, where $\lambda_{min}(.)$ are minimum eigen values of $A_{\tau}$ and $A_{rm}$.
\subsection{Online Parameter Estimation law}
\label{Identification}
The estimate to unknown true network parameters $W^*$ are evaluated on-line using gradient descent algorithm; correcting the weight estimates in the direction of minimizing the instantaneous tracking error $e(t) = x(t)-\hat x(t)$.  The resulting update rule for network weights in estimating the total uncertainty in the system is as follows
\begin{equation}
\dot {\hat W} = \Gamma Proj(\hat W,\phi(x)e(t)'P) \label{eq:18}
\end{equation} 

\subsubsection{Lyapunov Analysis}
The on-line adaptive identification law (\ref{eq:18}) guarantees the asymptotic convergence of the observer tracking errors $e(t)$ and parameter error $\tilde W(t)$ under the condition of persistency of excitation \cite{aastrom2013adaptive,ioannou1988theory} for the structured uncertainty. Under the assumption of unstructured uncertainty, we show tracking error is uniformly ultimately bounded (UUB).

\begin{theorem}
    Consider the actual and observer plant model (\ref{eq:unified uncertainty dynamics}) \& (\ref{eq:7}). If the weights parameterizing total uncertainty in the system are updated according to identification law \eqref{eq:18}, then the observer tracking error and error in network weights $\|e\|$, $\|\tilde W\|$ are bounded.
\label{Theorem-1}
\end{theorem} 

\begin{proof}
Let $V(e,\tilde W) > 0$ be a differentiable, positive definite radially unbounded Lyapunov candidate function,
%ensuring the asymptotic convergence of tracking error $e(t)$ and error in network weight $\tilde W$
\begin{equation}
V(e,\tilde W) = e^TPe + \frac{\tilde W^T \Gamma^{-1} \tilde W}{2}
\label{eq:20}
\end{equation}
where $\Gamma >0$ is the adaption rate. The time derivative of the lyapunov function (\ref{eq:20}) along the trajectory (\ref{eq:15}) can be evaluated as
\begin{equation}
\dot V(e,\tilde W) = \dot e^TPe + e^TP \dot e - \tilde W^T\Gamma^{-1}\dot{\hat W}
\label{eq:25}
\end{equation}
\begin{eqnarray}
\dot V(e,\tilde W) &=& -e^TQe + \left(\phi(x)e'P - \Gamma^{-1}\dot{\hat W}\right)2\tilde W \nonumber\\
&& + 2e^TP\epsilon(x) 
\label{eq:21}
\end{eqnarray}
for $P = P^T >0$ and $A_{\tau}$ be Hurwitz matrix, $P$ is the solution of lyapunov equation $A^T_{\tau}P + PA_{\tau} = -Q$ for some $Q > 0$.
 
Using the expressions for weight update rule (\ref{eq:18}) in (\ref{eq:21}), the time derivative of the lyapunov function reduces to 
\begin{eqnarray}
\dot V(e,\tilde W) &=& -e^TQe + 2e^TP\epsilon(x)
\label{eq:22}
\end{eqnarray}
\begin{eqnarray}
\dot V(e,\tilde W) &\leq& -\lambda_{min}(Q)e^Te + 2\lambda_{max}(P)\bar\epsilon e 
\label{eq:23}
\end{eqnarray}
Hence $\dot V(e,\tilde W) \leq 0$ outside compact neighborhood of the origin $e = 0$, for some sufficiently large $\lambda_{min}(Q)$. 
\begin{equation}
\|e(t)\| \geq \frac{2\lambda_{max}(P)\bar\epsilon}{\lambda_{min}(Q)}
\label{eq:error_bound}
\end{equation}
% Hence the observer tracking tracking error $\|e(t)\|$ is uniformly ultimately bounded. Furthermore, from the BIBO assumption $x_{rm}(t)$ is bounded for bounded reference signal $r(t)$, thereby $x(t)$ remains bounded.  Since $V(e,\tilde W)$  is radially unbounded the result holds for all $x(0) \in \mathcal{D}_x$. The boundedness of estimation parameters follows from Lyapunov theory and Barbalat’s Lemma \cite{narendra2012stable}
Hence the tracking tracking error $\|e(t)\|$ is uniformly lower bounded. Furthermore, from the BIBO assumption $x_{rm}(t)$ is bounded for bounded reference signal $r(t)$, thereby $x(t)$ remains bounded.  Since $V(e,\tilde W)$  is radially unbounded the result holds for all $x(0) \in \mathcal{D}_x$. We note that, the second derivative of Lyapunov function
\begin{equation}
\ddot{V}(e,\tilde W) = -2\lambda_{min}(Q)(e \dot{e}) + 2\lambda_{max}(P)\bar\epsilon \dot{e}
\end{equation}
is bounded due the fact that $\tilde{W}$ is bounded through projection operator in weight update rule \cite{ioannou1996robust} and $\bar{\epsilon}$ is finite constant, hence 
from Lyapunov theory and Barbalat’s Lemma \cite{narendra2012stable}, we can state that $\dot{V}(e,\tilde{W})$ is uniformly continuous hence $\dot{V}(e,\tilde{W}) \to 0$ as $t \to \infty$.  Using the previous fact with lower bound on error (\ref{eq:error_bound}) we show that $e(t)$ is uniformly ultimately bounded near to zero solution.
\end{proof}

%% file: Hybrid_MRAC.tex
\section{Hybrid Model Reference Adaptive Controller}
\label{reference_model}
This section provides the details of the hybrid architecture of the controller for the system with uncertainties. The total controller architecture uses both direct and indirect controller approach to handle system with matched and unmatched uncertainties.

\subsection{Indirect Adaptive Controller}
\label{Indirect}
An indirect adaptive control approach uses the estimate of unknown parameters or uncertainties of the plant through a companion observer model (\ref{eq:6}) to adjust the control parameters to achieve desired reference model tracking \cite{aastrom2013adaptive}.

Subjected to assumptions of boundedness and non-destabilizing effects; the unmatched uncertainties can be addressed through indirect adaptive feed-back $(K^\sigma(t)x(t))$ and feed-forward controller$(K_r^\sigma r(t))$ designed for an augmented plant model. The unmatched uncertainty estimates can be remodeled into an equivalent state dependent linear form, augmenting the nominal plant, for which the controller $K^\sigma(t), K_r^\sigma$ are designed. This architecture leads to global performance in reference model tracking to the original nonlinear system. 
  
\subsection{Direct Adaptive Controller}
\label{direct}
The direct adaptive controller $\nu_{ad}$ aims at canceling the matched uncertainties, and to minimize the reference model tracking error by ensuring $x(t) \to x_{rm}(t)$.

The total controller comprises of both indirect controller term $u_{pd} = K^\sigma (t)x(t)$, a feed-forward term $u_{rm} = K_r^\sigma r(t)$ and an direct adaptive element $\nu_{ad}$. The total controller $u(t)$ can be written as 
\begin{equation}
u(t) = u_{pd}+u_{rm}-\nu_{ad}
\label{eq:Adaptive Control}
\end{equation}
\begin{equation}
\text{where                    }\nu_{ad} = \hat \Delta_m(x)
\end{equation}
substituting the controller $u(t)$ in plant model (\ref{eq:1}) we get
\begin{small}
    \begin{eqnarray}
    \dot {x}(t)  &=& Ax(t) + \mathcal{B}\left(K^\sigma(t)x(t) + K_r^\sigma r(t)-\nu_{ad} + \Delta_m(x)\right) \nonumber \\
    &&+\mathcal{B}_u\Delta_u(x)
    % \dot {x}(t)  &=& \left(A + \mathcal{B}K^\sigma(t)\right)x(t) + \mathcal{B}K_r^\sigma r(t) + \mathcal{B}\left(\Delta_m(x) - \nu_{ad}\right) \nonumber \\
    % &&+\mathcal{B}_u\Delta_u(x)
    \label{eq:26}
    \end{eqnarray}
\end{small}
The feedback and feed-forward gain $K^\sigma (t)$ and $K_r^\sigma$ are designed to ensure the matching condition holds under switching system. Such that the closed loop poles of the system match the poles of reference model,
\begin{eqnarray}
\lambda(A_{rm}) &=& \lambda(\Psi_{\sigma}(\hat x)+\mathcal{B}K^\sigma(t)), \hspace{3mm} \forall t \label{eq:matching_condition1} \\
\mathcal{B}_{rm} &=& \mathcal{B}K_r^\sigma
\label{eq:matching_condition2}
\end{eqnarray} 
where
$$\Psi_{\sigma}(\hat x) = \begin{cases}
A \hspace{16mm} \sigma(t) \geq \gamma\\ 
A+A'(\hat x) \hspace{3mm} \sigma(t) \leq \gamma
\end{cases}$$

Where $A'(x) \in \mathbb{R}^{n \times n}$ is the state-dependent coefficient (SDC)  form of the unmatched uncertainty $\mathcal{B}_u\Delta_u(x)$. The details of remodelling unmatched uncertainty into SDC form is presented in further sections.

Pole placement method is used to calculate the feedback gain $K^\sigma(t)$, for which we require the pair $(\Psi_{\sigma}, \mathcal{B})$ be controllable. We prove the pair is always controllable in the Theorem-4. 

The feed-forward gain $K_r^\sigma$ is calculated using the expression
\begin{equation}
K_r^\sigma = (\mathcal{B}^T\mathcal{B})^{-1}\mathcal{B}^T\mathcal{B}_{rm}
\label{eq:feedforward_gain}
\end{equation}
For the solution of expression (\ref{eq:feedforward_gain}) to exist, we need $\mathcal{B} \in \mathbb{R}^{n \times m}$ to be full column rank.
For the case when $m \leq n$, this condition on $\mathcal{B}$ is usually satisfied for most dynamical systems. This assumption can be restrictive when $m > n$, in such cases we can use the pseudo-inverse approach to generate $K_r^\sigma$, such that the matching condition holds \cite{gao1992reconfigurable}.

\textbf{Condition 1} The feedback gain $K^\sigma(t)$ and hence reference model $A_{rm} = A+\mathcal{B}K^\sigma(t)$ is chosen such that $A_{rm}$ is Hurwitz and system is robust to effects of unmatched uncertainty (before the switching condition is triggered). This condition ensures the tracking error remain bounded before the controller switches to cater to the augmented model.

From \textbf{Assumption 1} the matched and unmatched uncertainty components can be upper bounded as,
$$\left[\begin{array}{c}
\bar{\xi}_m \\ 
\bar{\xi}_u
\end{array} \right] = \Omega^{-1}\mathcal{W}^T\phi(x) $$ 

\textbf{Remark 2}
For the system \eqref{eq:26} and given upper bounded $\|\Delta_u(x)\|_\infty \leq \bar{\xi}_u$, a suitable reference model can be selected using sector bounds and Linear Matrix Inequalities \cite{yang2010lmi} from robust control theory to ensure \textbf{Condition 1} holds.\\

\begin{theorem}
Given the actual and reference plant model \eqref{eq:1} \& \eqref{eq:ref model} respectively. We show that the controller of the form \eqref{eq:Adaptive Control} is the admissible controller for the system \eqref{eq:1} and the reference model tracking error $\|e_{rm}\|$ is uniformly ultimately bounded(UUB), 
\end{theorem}
\vspace{2mm}
\begin{proof}
Let the reference model tracking error $e_{rm}$  is defined as:
\begin{equation}
e_{rm} = x(t) - x_{rm}(t)
\label{eq:ref_model_tracking_err}
\end{equation}
Taking time derivative of \eqref{eq:ref_model_tracking_err} and using \eqref{eq:1} \& \eqref{eq:ref model}, the reference model tracking error rate can be written as
\begin{eqnarray}
\dot e_{rm} &=& Ax(t) +  \mathcal{B}(u(t) +\Delta_m(x))+\mathcal{B}_u\Delta_u(x) \nonumber \\
&&- A_{rm}x_{rm} - \mathcal{B}_{rm}r(t)
\end{eqnarray}
Using the controller \eqref{eq:Adaptive Control} and under the assumption of switched systems, the above equation can be written as,
\begin{equation}
\dot e_{rm} = \begin{cases}
A_{rm}e_{rm} + \mathcal{B}(\Delta_m -\nu_{ad}) + \mathcal{B}_u\Delta_u, \hspace{5mm} \sigma(t) \geq \gamma\\ 

A_{rm}e_{rm} + \mathcal{B}(\Delta_m-\nu_{ad}), \hspace{18mm} \sigma(t) \leq \gamma
\end{cases}
\label{eq:ref_error-_dynamics}
\end{equation}
Let $V(e_{rm}) > 0$ be a differentiable, positive definite radially unbounded Lyapunov candidate function,
%ensuring the asymptotic convergence of tracking error $e(t)$ and error in network weight $\tilde W$
\begin{equation}
V(e_{rm}) = e_{rm}^TPe_{rm}
\label{eq:ref_model_lyapunov}
\end{equation}
The time derivative of the lyapunov function (\ref{eq:ref_model_lyapunov}) along the trajectories (\ref{eq:ref_error-_dynamics}) can be evaluated as
\begin{equation}
\dot V(e_{rm}) = \dot e_{rm}^TPe_{rm} + e_{rm}^TP \dot e_{rm}
\end{equation}
for $P = P^T >0$ and Hurwitz matrix $A_{rm}$, $P$ is the solution of lyapunov equation $A^T_{rm}P + PA_{rm} = -Q$ for some $Q > 0$.

Hence $\dot V(e_{rm}) \leq 0$ outside compact neighborhood of the origin $e_{rm} = 0$, for some sufficiently large $\lambda_{min}(Q)$. 
\begin{equation}
\|e_{rm}(t)\|_\infty \geqslant \begin{cases}
\frac{\mathcal{B} \bar \xi_m + \mathcal{B}_u \bar \xi_u}{\lambda_{min}(Q)} \hspace{13mm} \sigma(t) \geq \gamma\\ 
\\
\frac{\mathcal{B}\bar \xi_m}{\lambda_{min}(Q)}  \hspace{16mm} \sigma(t) \leq \gamma
\end{cases}
\end{equation}
Using similar argument from Theorem-\ref{Theorem-1} and using Barbalat’s Lemma \cite{narendra2012stable} we show reference model error is UUB and this bound can be made sufficiently small by suitably choosing a faster reference model.
\end{proof}
\subsection{Switched Systems}
The observer model tracking error $e(t)$ is a direct measure of how accurate are the available instantaneous estimates of total uncertainty in the system. Hence the norm on the observer tracking error 
\begin{equation}
\sigma(t) = \|e(t)\|_2 \triangleq \|x(t)-\hat x(t)\|_2 
\end{equation}
is used as switching signal $\sigma(t)$, between nominal system and the augmented system. 

Algorithm-1 provides details of the switching approach. Tracking error $\sigma(t) = \|e(t)\|_2$ indicates the switching between nominal and augmented plant model. In step-4 algorithm tests if the $\|e(t)\|_2 \leq \gamma$,  this condition indicates we have a fairly accurate model of the actual dynamics and that the estimates of the uncertainties are close to their true values. In step-5, the unmatched part of the uncertainty is remodeled in state-dependent form. The SDC form $A'(x) \in \mathbb{R}^{n \times n}$ while capturing all the nonlinearities, represents the uncertainty in a non-unique linear structure \cite{cloutier1997state}. The SDC form of unmatched uncertainty estimate can be expressed as,
\begin{equation}
A'(x)x(t) = \mathcal{B}_u\hat \Delta_u(x)
\end{equation}
\begin{algorithm}[h!]
    \caption{Nominal Model Switch}
    \label{alg:Apprentice}
    \begin{algorithmic}[1]
        \STATE {\bfseries Input:} $\gamma$
        \STATE while time $t < \infty$
        \STATE $\|e(t)\| = \|x(t)-\hat x(t)\|$
        \IF{$\|e(t)\| \leq \gamma$}
        \STATE Model the unmatched uncertainty into equivalent SDC form $$\mathcal{B}_u\hat \Delta_u(x) = A'(x)x(t)$$
        \STATE Augment the nominal plant model $$\Psi_{\sigma}(x) = A_\tau + A'(x)$$
        \ENDIF
        \end{algorithmic}
\end{algorithm}
The SDC form of unmatched uncertainty augments nominal system $A$ to form the complete nonlinear model,
\begin{eqnarray}
\dot {\hat x}(t)  &=& \left(A + A'(\hat x)\right)\hat x(t) + B\left(u(t) + \hat \Delta_m(x)\right) \nonumber \\
&& + L_{\tau}(x(t) - \hat x(t)) \label{eq:26_1} \\
\dot {\hat x}(t)  &=& \Psi_{\sigma}(\hat x)\hat x(t) + B\left(u(t) + \hat \Delta_m(x)\right) \nonumber \\
&& + L_{\tau}(x(t) - \hat x(t))
\label{eq:26_2}
\end{eqnarray}
%\begin{equation}
%\dot {\hat x}(t)  = \Psi_{\sigma}(\hat x)\hat x(t) + B\left(u(t) + \hat \Delta_m(x)\right)
%\label{eq:26_2}
%\end{equation}
where $\Psi_{\sigma}(\hat x) \in \mathbb{R}^{n \times n} $ is the total SDC augmented system dynamics matrix.
\begin{equation}
\Psi_{\sigma}(\hat x) =  \left(A + A'(\hat x)\right)
\label{eq:augmented_model}
\end{equation}
As long as the condition holds $\sigma(t) \leq \gamma$ the augmented system dynamics \eqref{eq:augmented_model} is used for indirect adaptive control design for the system.

We operate in the realm of switched linear systems \cite{hespanha2004uniform} where $A_{\sigma} = (\Psi_{\sigma} + \mathcal{B}K^\sigma(t))$ belongs to set of finitely many Hurwitz matrices i.e. $A_{\sigma} \in \{A_1,A_2 \ldots\}$. Also from the matching condition \eqref{eq:matching_condition1}, we know that $A_{\sigma}$ belongs to family of matrices which share common eigen values. The stability of such a arbitrarily switched linear system is given by the following theorem.
\begin{theorem}
For switched linear system of form
\begin{equation}
\dot x(t) = A_{\sigma}x(t) + \mathcal{B}r(t)
\label{eq:switched_sys}
\end{equation}
If there exists a quadratic lyapunov function $V_{\sigma} = x(t)^TP_{\sigma}x(t)$,  such that, $P_{\sigma}$ belongs to a compact family $P_{\sigma} \in \mathcal{P} \in \mathbb{R}^{n \times n}$ of symmetric positive definite matrices, Such that, for every $x(t)$, $P_{\sigma}$ satisfies
$$A^T_{\sigma}P_{\sigma} + P_{\sigma}A_{\sigma} \leq -Q \hspace{3mm} \forall \sigma \textrm{ and } Q > 0$$
and the lie derivative $\dot V_{\sigma} < 0$ along \eqref{eq:switched_sys}, then system \eqref{eq:switched_sys} is said to be uniformly asymptotically stable.
\end{theorem}
\begin{proof}
The proof of the following theorem is provided in \cite{hespanha2004uniform}
\end{proof}

\subsection{SDC Formulation of Unmatched uncertainty}
\label{SDC}
Upon satisfying the switching condition $(\sigma(t) \leq \gamma)$, the indirect adaptive control remodels the unmatched uncertainty estimate to state dependent matrix form  $A'(x) \in \mathbb{R}^{n \times n}$, such that $A'(x)$ augments nominal system matrix $A$ for the indirect adaptive controller design.

To represent the estimate of unmatched uncertainty in SDC form $A'(x)x(t)$ i.e.
\begin{equation}
A'(x)x(t) = \mathcal{B}_u\hat \Delta_u(x)
\label{eq:unmateched_term}
\end{equation}
Multiply and divide the term on RHS of (\ref{eq:unmateched_term}) by the term $(x(t)^Tx(t))$. To mitigate the problem of a zero in the denominator, we add a Tikhonov regularization term \cite{bauer2007regularization,Guacaneme1988} to the denominator of (\ref{eq:unmateched_term2}) and the resulting expression is as follows
\begin{equation}
A'(x)x(t) \approx \frac{\mathcal{B}_u\Delta_u(x)}{(x^Tx+\epsilon)}x^Tx
\label{eq:unmateched_term2}
\end{equation}
Equation (\ref{eq:unmateched_term2}) can be rearranged in its final SDC form as,
\begin{eqnarray}
A'(x)x(t) \triangleq \mathcal{B}_u\hat \Delta_u(x) &\approx&   \left[\frac{x\Delta^T_u(x)\mathcal{B}_u^T}{(x^Tx+\epsilon)}\right]x(t)\\
&& \nonumber \\
A'(x) &=& \left[\frac{x\Delta^T_u(x)\mathcal{B}_u^T}{(x^Tx+\epsilon)}\right]
\label{eq:unmateched_term3}
\end{eqnarray}
%The systems model is switched from $A$ to $A+A'(x)$ as the tracking error satisfies the condition, $\|e\| \leq \gamma$\\

\begin{theorem}
    Given the pair of system matrices $\left(A,\mathcal{B}\right)$ is controllable, and if the non-unique SDC form of unmatched uncertainty is modeled as
    $$A'(x) = \left[\frac{x\Delta^T_u(x)\mathcal{B}_u^T}{(x^Tx+\epsilon)}\right]$$
    Then the augmented system matrix pair $\left(A+A'(x),\mathcal{B}\right)$ is always controllable,
\end{theorem} 
\begin{proof}
Writing the controllability matrix for the augmented system,
\begin{equation}
\dot x(t) = (A+A'(x))x(t) + \mathcal{B}u(t)
\end{equation}
\begin{equation}
\mathcal{C} = \left[\mathcal{B}|\left(A+A'(x)\right)\mathcal{B}|\ldots |\left(A+A'(x)\right)^{n-1}\mathcal{B}\right]
\end{equation}
For the pair $\left(A+A'(x),\mathcal{B}\right)$ to be controllable $Rank(\mathcal{C}) = n$. Therefore controllability matrix $\mathcal{C}$ should contain $n$ independent columns.

We can establish this claim by proving mutual independence of the columns of the controllability matrices, term by term
\[\left(A+A'(x)\right)\mathcal{B} = A\mathcal{B}+A'(x)\mathcal{B}\] 
Evaluating the term $A'(x)B$
\begin{equation}
A'(x)\mathcal{B} = \frac{x(t)\hat \Delta^T_u(x)\mathcal{B}^T_uB}{\left(x(t)^Tx(t) + \epsilon\right)}
\label{eq:27}
\end{equation}
From \emph{Remark 1} we know that $\mathcal{B}^T_u\mathcal{B} = 0$ and therefore $A'(x)\mathcal{B} = 0$ and hence
\begin{equation}
\left(A+A'(x)\right)\mathcal{B} = A\mathcal{B}
\label{eq:28}
\end{equation}
To similarly show $\left(A+A'(x)\right)^n\mathcal{B} = A\mathcal{B}$, $\forall  n$, lets consider a term for any $n=k$
\[\left(A+A'(x)\right)^k\mathcal{B}\] 
Using the identity
\[\left(a+b\right)^k = {k \choose 0}a^k + {k \choose 1}a^{k-1}b + {k \choose 2}a^{k-2}b^2\ldots {k \choose k}b^k\]
The expansion of term $\left(A+A'(x)\right)^k$ can be written as, involving the terms $A'(x)$ raised to power from $0$ to $k$.
\begin{equation}
\left(A+A'(x)\right)^k = {k \choose 0}A^k+{k \choose 1}A^{k-1}A'(x) \ldots +{k \choose k}A'(x)^k
\label{eq:mathematical_induction1}
\end{equation}
and therefore
\begin{eqnarray}
\left(A+A'(x)\right)^k\mathcal{B} &=& {k \choose 0}A^k\mathcal{B}+ {k \choose 1}A^{k-1}A'(x)\mathcal{B}\nonumber \\
&&\ldots + {k \choose k}A'(x)^k\mathcal{B}
\label{eq:mathematical_induction2}
\end{eqnarray}
Lets considering one of the term in the above identity,  ${k \choose m}A^{k-m}A'(x)^m\mathcal{B}$ in the series expansion (\ref{eq:mathematical_induction2}) to evaluate the identity $\left(A+A'(x)\right)^k\mathcal{B}$
\begin{equation}
{k \choose m} A^{k-m}A'(x)^m\mathcal{B} = {k\choose m} A^{k-m}A'(x)^{m-1}A'(x)\mathcal{B}
\end{equation}
From (\ref{eq:27}) and \emph{Remark 1} we know that $A'(x)\mathcal{B} = 0$ and therefore can be generalized for any ``$m$", i.e. $A'(x)^m\mathcal{B} = 0, \forall m$
and hence we can show that,
\begin{equation}
\left(A+A'(x)\right)^k\mathcal{B} = A^k\mathcal{B}
\label{eq:29}
\end{equation} 
and therefore by mathematical induction, (\ref{eq:29}) is true for any $k = n$, hence the controllability matrix for augmented system is proved to be invariant, i.e.
\begin{eqnarray}
\mathcal{C}&=&\left[\mathcal{B}|\left(A+A'(x)\right)\mathcal{B}|\ldots |\left(A+A'(x)\right)^{n-1}\mathcal{B}\right] \nonumber \\
&=& \left[\mathcal{B}|A\mathcal{B}|\ldots |\left(A\right)^{n-1}\mathcal{B}\right]
\end{eqnarray} 
And since the original system $\left(A,\mathcal{B}\right)$ is controllable $Rank(\mathcal{C}) = n$, the augmented system $A+A'(x)$ with $A'(x)$ of form (\ref{eq:unmateched_term3}) is always controllable.\\    
\end{proof}

%% file: Simulations.tex
\section{Simulations}
\label{results}
In this section, we evaluate the presented Hybrid Direct-Indirect adaptive control by numerical simulations on a nonlinear plant.
Consider a nonlinear dynamical system,
\begin{eqnarray}
\dot x_1(t) &=& x_2(t)-x_1(t) \nonumber \\
\dot x_2(t) &=& 0.5x_1(t)-x_2(t)-x_1(t)x_3(t) \label{eq:example_plant}\\
\dot x_3(t) &=& x_1(t)x_2(t)-x_3(t)+u(t) \nonumber
\end{eqnarray}
The nonlinear system (\ref{eq:example_plant}) can be written in terms of nominal system dynamics and total model uncertainty of form \eqref{eq:1}.

The nominal system dynamics is selected to be linear and such that
the pair $\left(A,\mathcal{B}\right)$ is controllable. The total uncertainty $\Delta(x)$ can be written as sum of two unknown nonlinear functions $\Delta_m(x)$ and $\Delta_u(x)$ belonging to range space and null space of $\mathcal{B}$, 
\begin{eqnarray}
\dot x(t) &=& \left[\begin{array}{ccc}
-1&1&0  \\ 
0.5&-1&1\\
0&0&1
\end{array} \right]x(t)+\left[\begin{array}{c}
0\\ 
0\\
1
\end{array}\right](u(t)+x_1x_2) \nonumber\\
&&+\left[\begin{array}{cc}
0&1\\ 
1&0\\
0&0
\end{array}\right]\left[\begin{array}{c}
x_3-x_1x_3\\ 
x_1x_2
\end{array} \right]
\end{eqnarray}
where $x(t) = \left[\begin{array}{ccc}
x_1 & x_2 & x_3
\end{array} \right]^T$ is state of system (\ref{eq:example_plant})

The matched and unmatched uncertainty in the system are represented as follows,
\begin{equation}
\Delta_m(x) = x_1x_2, \hspace{2mm}\Delta_{u}(x) = \left[\begin{array}{c}
x_3-x_1x_3\\ 
x_1x_2
\end{array} \right]\nonumber
\end{equation}

Initial conditions for the simulation are arbitrarily chosen to be $x(0) = [0\; 0 \;0]^T$. The reference model chosen is a stable third order linear system with eigen values $\lambda(A_{rm}) = [-3,-4,-5]^T$. The linear control gain $K^\sigma(t),K_r^\sigma$ are evaluated using pole placement and pseudo-inverse methods. 
The simulation runs for a total time of 120 seconds with an update rate $dt = 0.05$ seconds using RK-4 integration. The threshold error bound for switching the system from $A$ to augmented model $A+A'(x)$ is $\gamma = 0.001$ and the learning rate is set to $\Gamma  = 0.05$. We use a Radial Basis Function(RBF) network for uncertainty estimation with $10$ centers selected in range $[-1\;\;1]$ with bandwidth $\sigma = 0.25$ to have sufficient overlap between RBF activation function to allow smooth function learning.

The reference model tracking performance of the Hybrid Direct-Indirect MRAC algorithm is shown in Fig-\ref{fig:plot_1}. The switching signal between nominal and the augmented plant model is shown in Fig-\ref{fig:plot_1}, switch-$0$ indicate $e(t) \geq \gamma$, and therefore the nominal model is used for the indirect controller, and switch-$1$ indicates $e(t) \leq \gamma$ and hence the augment model is used for controller synthesis. Figure-\ref{fig:plot_2} provide the performance of adaptation law in approximating the matched and un-matched uncertainty. Both matched, and unmatched uncertainty estimates are satisfactorily close to their true values. The performance of the controller in reference tracking and observer model tracking is shown through state error plot in Fig-\ref{fig:plot_4}. Figure-\ref{fig:plot_3} show the evolution of weights in approximating total uncertainty. The control $u(t)$ for reference model tracking for the system with unmatched uncertainty is shown in Fig-\ref{fig:plot_5}. It can be observed that the proposed controller is able to control the plant under matched and unmatched uncertainties and achieve tracking of the designed reference model.
\begin{figure}
\centering
\includegraphics[width=1\linewidth]{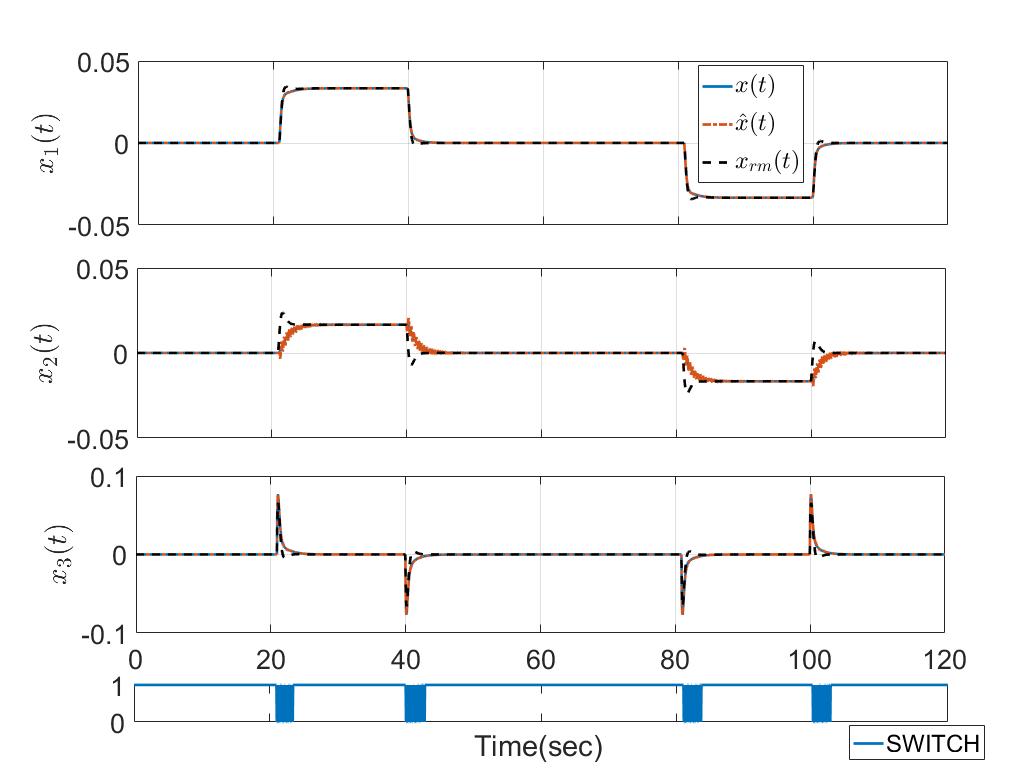}
\caption{Reference model tracking performance of the Hybrid MRAC
    adaptive controller \& Signal for switching between Nominal plant $(0)$ and Augmented plant $(1)$}
\label{fig:plot_1}
\end{figure}
\begin{figure}
\centering
\includegraphics[width=1\linewidth]{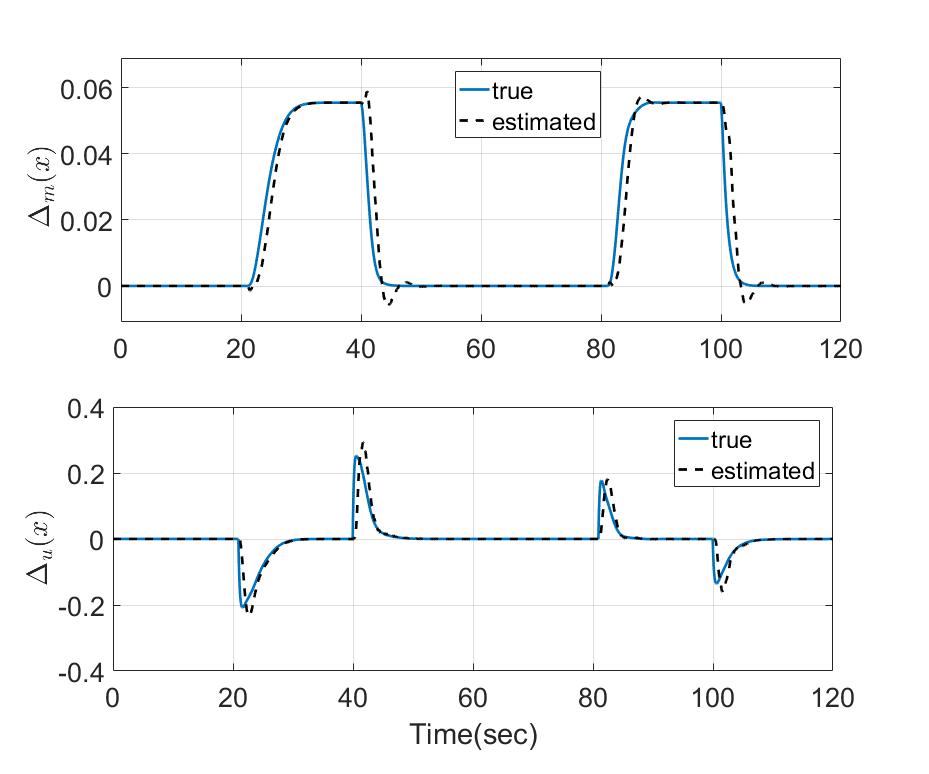}
\caption{Matched Uncertainty captured using online Adaption law using Hybrid MRAC}
\label{fig:plot_2}
\end{figure}
\begin{figure}
    \centering
    \includegraphics[width=1\linewidth]{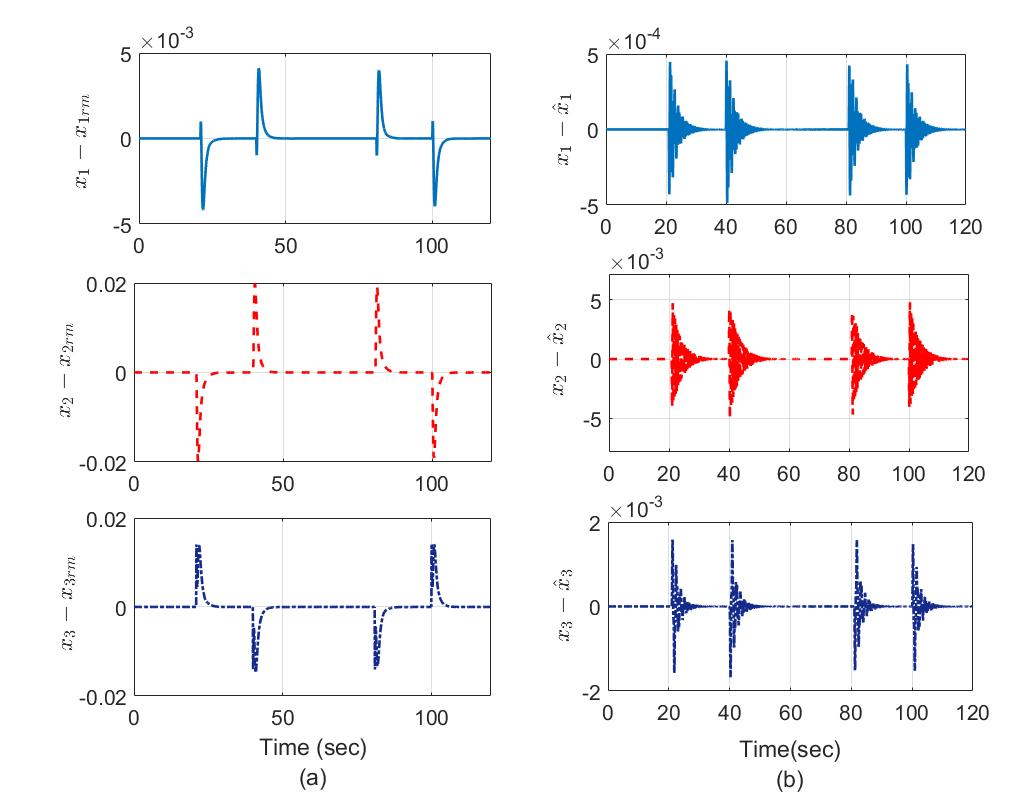}
    \caption{Tracking Error: (a) Reference Model tracking error (b) Observer Model tracking error}
    \label{fig:plot_4}
\end{figure}
\begin{figure}
    \centering
    \includegraphics[width=0.8\linewidth]{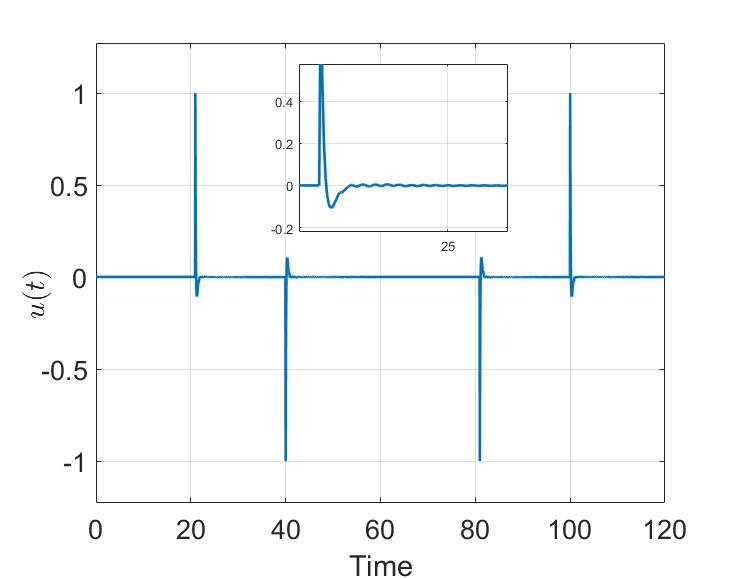}
    \caption{Control input u(t) for the desired reference model tracking}
    \label{fig:plot_5}
\end{figure}
\begin{figure}
    \centering
    \includegraphics[width=0.8\linewidth]{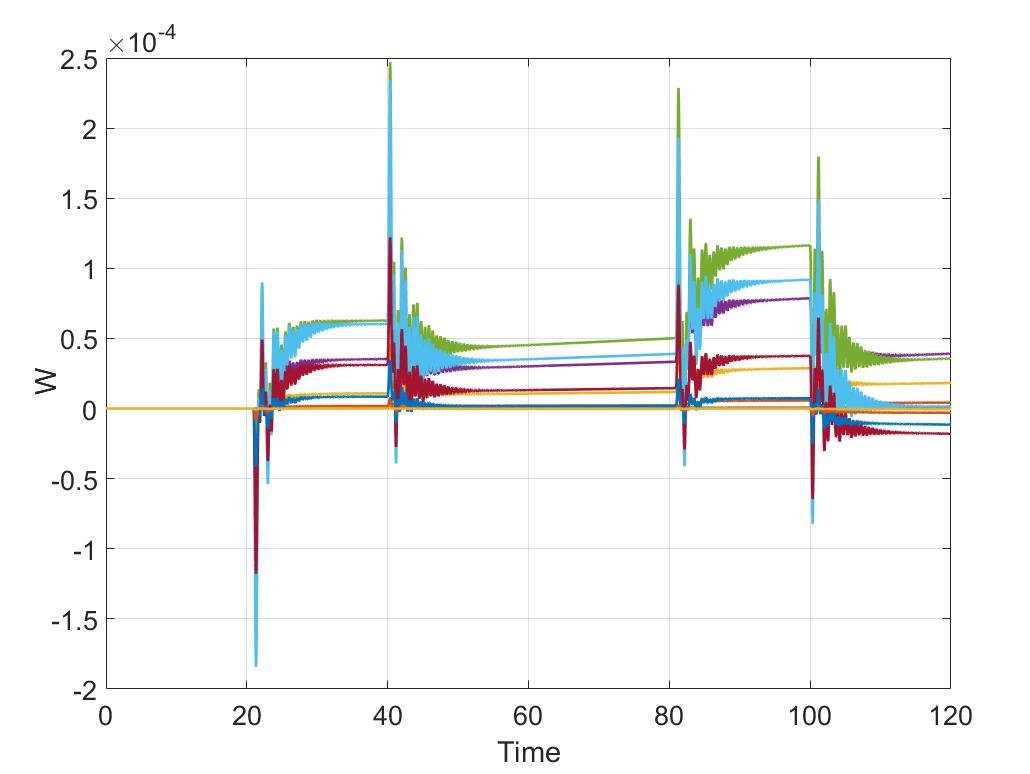}
    \caption{RBF Network Weights for uncertainty approximation}
    \label{fig:plot_3}
\end{figure}

%% file: conclusion.tex
\section{Conclusion}
\label{conclusions}
In this paper, we presented a hybrid direct-indirect adaptive control using MRAC architecture to address a class of problems with matched and unmatched uncertainties. The proposed controller uses an observer model to estimate the unmatched uncertainty and use this information in indirect control synthesis. The estimation of unmatched uncertainty is shown to be possible with observer tracking error alone. It is also demonstrated that the presented method is not restricted to only linear in state uncertainties. However, a broad class of unstructured nonlinear uncertainties can be handled in both matched and unmatched part. We have shown the existence of guaranteed uniform bounds on tracking error under switching systems. Numerical simulations with a non-linear plant model demonstrate the controller performance, in achieving reference model tracking in the presence of matched and unmatched uncertainties.